\documentclass[letterpaper, 10 pt, conference]{ieeeconf}  % Comment this line out
\pdfoutput=1
\makeatletter

\let\proof\@undefined
\let\endproof\@undefined

\let\theorem\@undefined
\let\endtheorem\@undefined
\makeatother

\IEEEoverridecommandlockouts                              % This command is only
\usepackage{times}
\usepackage{helvet}
\usepackage{courier}
\usepackage{srcltx} % SRC Specials: DVI [Inverse] Search
\usepackage{amsmath,amsthm,amsxtra,amssymb}
\usepackage{times,latexsym,array,verbatim,url}
\usepackage{epsfig,color,graphicx,graphics}
\usepackage{mdwmath,mdwtab,mathptmx,cite}
\usepackage{xspace} % for adaptive spacing eg if followed by '.' don't space

\newtheorem{thm}{Theorem}

\newtheorem{lem}[thm]{Lemma}
\newtheorem{assm}{Assumption}
\newtheorem{prop}[thm]{Proposition}

\newtheorem{defn}{Definition}
\theoremstyle{remark}

\newcommand{\norm}[1]{\left\Vert#1\right\Vert}
\newcommand{\abs}[1]{\left\vert#1\right\vert}

\newcommand{\Real}{\mathbb R}
\newcommand{\eps}{\varepsilon}

\newcommand{\mc}{\mathcal}

\newcommand{\g}{\gamma}

\renewcommand{\a}{\alpha}

\newcommand{\one}{\mathbf{1}}

\newcommand{\ie}{\emph{i.e.,}\xspace}
\newcommand{\eg}{\emph{e.g.,}\xspace}
\newcommand{\cf}{\emph{cf.}\xspace}

\title{\LARGE \bf
Large-Scale Strategic Games and Adversarial Machine Learning
}

\author{Tansu Alpcan, Benjamin I. P. Rubinstein, and Christopher Leckie% <-this % stops a space
\thanks{This work was supported in part by the ARC Discovery Project DP140100819, DECRA DE160100584 and grant FLI-RFP-AI1.}% <-this % stops a space
\thanks{T. Alpcan {\tt\small tansu.alpcan@unimelb.edu.au} is with Department of Electrical and Electronic Engineering; B. Rubinstein {\tt\small benjamin.rubinstein@unimelb.edu.au} and C. Leckie 
{\tt\small caleckie@unimelb.edu.au} are with Department of Computing and Information Systems, The University of Melbourne, Australia. }%
}

\begin{document}

\maketitle
\thispagestyle{empty}
\pagestyle{empty}

%%%%%%%%%%%%%%%%%%%%%%%%%%%%%%%%%%%%%%%%%%%%%%%%%%%%%%%%%%%%%%%%%%%%%%%%%%%%%%%%
\begin{abstract}
Decision making in modern large-scale and complex systems such as
communication networks, smart electricity grids, and cyber-physical
systems motivate novel game-theoretic approaches. This paper
investigates big strategic (non-cooperative)
games where a finite number of individual players each have a large
number of continuous decision variables and input data points.
Such high-dimensional decision spaces and big data sets 
lead to computational challenges, relating to efforts
in non-linear optimization scaling up to large systems of variables.
In addition to these computational challenges, real-world players
often have limited information about their preference parameters
due to the prohibitive cost of identifying them or due to operating
in dynamic online settings. The challenge of limited information is
exacerbated in high dimensions and big data sets. Motivated by both computational
and information limitations that constrain the direct solution of big
strategic games, our investigation centers around reductions using
linear transformations such as random projection methods 
and their effect on Nash equilibrium solutions. Specific
analytical results are presented for quadratic games and approximations. In addition,
an adversarial learning game is presented where random projection and sampling
schemes are investigated.
% Decision making in modern large-scale and complex systems such as
% communication networks, smart electricity grids, and cyber-physical
% systems motivate novel game-theoretic approaches. This paper
% investigates large-scale strategic (non-cooperative)
% games where a finite number of individual players each have a large
% number of continuous decision variables and input data points.
% Such high-dimensional decision spaces and large-scale data sets 
% lead to computational challenges, relating to efforts
% in non-linear optimization scaling up to large systems of variables.
% In addition to these computational challenges, real-world players
% often have limited information about their preference parameters
% due to the prohibitive cost of identifying them or due to operating
% in dynamic online settings. The challenge of limited information is
% exacerbated in high dimensions and large-scale data sets. Motivated by both computational
% and information limitations that constrain the direct solution of large-scale
% strategic games, our investigation centers around reductions using
% linear transformations such as random projection methods 
% and their effect on Nash equilibrium solutions. Specific
% analytical results are presented for quadratic games and approximations. In addition,
% an adversarial learning game is studied where random sampling
% schemes and their impact on solution is investigated analytically and numerically.
\end{abstract}

%%%%%%%%%%%%%%%%%%%%%%%%%%%%%%%%%%%%%%%%%%%%%%%%%%%%%%%%%%%%%%%%%%%%%%%%%%%%%%%%
\section{INTRODUCTION}
\noindent Game theory, which has its roots in economics, has recently become a mainstream approach to
a multitude of engineering problems in communications~\cite{mechbook}, %,yunus1}
electricity grids~\cite{debsawind,walidsmartgridgame,lawtansurisk}, and network security~\cite{alpcan-book,tambegame,Barth-2011-reactive}.
Providing a solid mathematical foundation for multi-agent decision making, game theory has also been
used extensively in optimization and control of networked and cyber-physical systems~\cite{basarcyberphysical1}. 

With the advent of large-scale data analytics, large-scale decision problems have become
prevalent in many engineering disciplines. Linear Programs with thousands of variables are now common-place; and convex optimization on large-scale data, aiming to overcome computational, storage and communication bottlenecks~\cite{convexbigdata1}, has emerged as a major area of study. %that relates closely to parallel and distributed computation~\cite{convexlarge-scaledata1}. 
Strategic games with continuous decision variables often rely on convex optimization and linear system theory. Unlike games with finite states and actions, studies on continuous-kernel games traditionally do not emphasize scalability nor large-scale data. \textit{In contrast to numerous works on game abstraction, there has been little discussion in the classical game theory literature on games with large numbers of continuous variables and big data sets.}

Strategic games with a large number of continuous variables and high-dimensional strategy spaces share unique research challenges with those involving large number of \textit{finite} states and/or actions~\cite{AAAI1510039}. One research challenge, common with convex optimization on large-scale data, is computational. Finding Nash Equilibrium strategies is often computationally prohibitive in large-scale games~\cite{algogame}. There have recently been efforts to address these computational issues using active-set 
or similar methods~\cite{tambelargegame1}. However, even novel computational methods may be infeasible in
certain scenarios such as repeated games played in real-time. The second and arguably more important challenge is a lack of information. It is difficult and oftentimes infeasible to identify the utility functions and preferences of individual players for each decision variable, especially if the number of variables grows to thousands or more.

\textit{This paper presents a framework for large-scale strategic \\ (non-cooperative) games where a finite number
individual players each have a large number of continuous decision variables.} Hence, it can be seen as complementary
to the existing literature on \textit{game abstraction} which shares similar aims~\cite{AAAI1510039}. The main difference is the focus on continuous-kernel games and strategy/action spaces. Specifically, nonzero-sum large-scale strategic games with high-dimensional continuous decision spaces and \textit{random projection} methods are investigated as a starting point. 
%Our focus is on information-limitations which impose hard limitations on solving the original large-scale game even when computationally feasible to do so. 
Our investigation centers around the reduction of large-scale strategic games using transformations such as random projections and their effect on Nash Equilibrium solutions. Analytically tractable results are presented for quadratic games and in an adversarial machine learning setting.

% With the advent of large-scale data analytics, large-scale decision
% \begin{itemize}
%  \item With the advent of large-scale data, large-scale decision problems have become prevalent.
%  \item Many optimisation methods are scalable, e.g. millions of variables in LPs
%  \item There is very little discussion in the game theory literature on games with large
%  number of decision variables.
%  \item this paper introduces a first step by focusing on large-scale strategic (non-cooperative)
%  games. Specifically, two player ZS, NZS and quadratic games are investigated.
% \end{itemize}
% 
% Games with large number of decision variables and high-dimensional strategy or decision spaces are often intractable due to
% \begin{enumerate}
%  \item lack of information: it is hard to identify the cost or utility functions (values)
%  taking into account each decision variable.
%  \item computation constraints: finding NE strategies
%  could be computationally very expensive in large-scale strategy spaces.
% \end{enumerate}
% 
% \textbf{Examples:} games in complex systems e.g. smart grid, adversarial learning, security games (e.g.
% in cloud computing), microeconomics, games on systems involving large-scale data.

\subsection{Related Work}

Games with a large number of players are well-known in the literature, \cf \eg the concept of Wardrop Equilibrium~\cite{wardrop1}. More recently, mean-field games have been studied under the assumption of very large numbers of players on large systems. The basic idea behind mean-field games is approximating 
large games by a stylized model with a continuum of players~\cite{tembinemf1}. 

Scalability issues arise in games from a number of different perspectives. For example, in cooperative decision-making among a large population of agents whose opinions must be considered, \cite{GreeneKLS09} propose a Bayesian belief aggregation scheme among many agents holding a potentially diverse range of opinions. \cite{MarcolinoXJTB14} examine how the performance of a team of diverse agents improves in cooperative games, as the number of agents or possible actions increases. 

Rather than cooperative games, this paper focuses on strategic (non-cooperative) games. In our setting, recent work has investigated how scalability issues arise in specific contexts. \cite{tambelargegame1} consider strategic games with sequential strategies, in which large search spaces arise due to strategies that require perfect memory of the history of play for ensuring the existence of equilibria. The authors propose a hybrid approach that combines a compact representation for strategies with an incremental approach to strategy generation in order to address the search space complexity. \cite{YangJTO13} address the problem of Stackelberg security games, in which defenders need to assign resources to protect targets against attackers. The major challenge that they focus on is finding defender strategies that satisfy the underlying constraints on the resources that need to be allocated to each strategy. The authors use a cutting-plane algorithm to speed up the search in the defender's solution space. 
%In contrast to this earlier work, this paper focuses on nonzero-sum large-scale strategic games with high-dimensional continuous decision spaces.

Game abstraction has emerged in recent years as a key enabler for solving large incomplete-information games with finite states and/or action sets. \cite{AAAI1510039} presents an excellent survey of abstraction of information or actions in games, motivated by incomplete information or scalability. Distinct to the literature on game abstraction, the focus of this paper is on games with continuous decision variables. 

Large-scale strategic games also arise in the context of adversarial machine learning: the study of statistical inference under adversarial influence.
A number of threat models fall under this umbrella~\cite{Barreno-2010}. While privacy-preserving learning has been met successfully by the theoretical frameworks of secure multi-party computation~\cite{lindell2000privacy} and differential privacy~\cite{Dwork06}, 
%Successful theoretical frameworks exist for learning on privacy-sensitive data: when releasing data to an untrusted learner via secure multi-party computation~\cite{lindell2000privacy}, or releasing the product of learning to untrusted third parties via %$k$-anonymity~\cite{kAnony} or the stronger notion of
%differential privacy~\cite{Dwork06}. 
%Secure multi-party computation applies homomorphic encryption to securely providing privacy-sensitive data to an untrusted learner. Differential privacy~\cite{Dwork06} and $k$-anonymity enable the release of models fit to privacy-sensitive data, to untrusted third parties. While cryptographic or information theoretic approaches serve these settings well,
less is known about how to learn or predict on poisoned data. The bulk of related work in computer security, has been on one-shot attacker strategies seeking to force classifier errors. %\cite{Barreno-2006} proposed a taxonomy of attacks on learners in terms of attacker influence and goals. Typical goals include affecting false positives or negatives (in the case of binary classification), which are in turn either targeted or indiscriminate in their effect on predictions.
%Attacker influence could be on test instances or extend to training data.
Related case studies include attacks on spam detection~\cite{wittel-wu-2004-attacking,lowd-meek-2005-good,LEET2008}, polymorphic worm detectors~\cite{newsome2006paragraph}, and on network anomaly detectors~\cite{fogla-lee-2006-evading,IMC09}.

Little progress has been made on defenses in adversarial learning. \cite{Dalvi-Domingos-KDD-2004} considered patching of simple classifier's `blind spots' to attack by instances that optimize an attacker cost function in a one-shot game-theoretic setting. \cite{Bruckner-Scheffer-NIPS-2009} identify conditions for the existence of unique Nash equilibria for static games for learning. However the models learned in both works are linear, representing a useful but limited model class. Regret minimizing learners~\cite{Cesa-Bianchi-Lugosi-ExpertGames-2006} have been used in security settings~\cite{Barth-2011-reactive,Blocki-RegretAudit-2011} but not with popular ``batch'' learners. Learning models in much larger classes such as random forests or deep neural networks are significantly high-dimensional, particularly in large-scale data settings where the data set size can support learning of large numbers of parameters in so-called non-parametric methods. These problems thus motivate the study of large-scale strategic games.

The adversarial learning setting relates to robust statistics~\cite{Huber-RS-1981} and online learning theory or the theory of regret minimization~\cite{Cesa-Bianchi-Lugosi-ExpertGames-2006}. The former assumes that an infinitesimal proportion of otherwise i.i.d. training data is contaminated by unbounded noise: not modeling a real-world attacker but rather a tool for studying break-down points and robustness to passive, benign outliers. Online learning theory comes closer to game theory, but the approach is typically to directly represent and update mixed strategies over predictions based on feedback in rounds. While there have been applications in security~\cite{IMC09,Barth-2011-reactive,Blocki-RegretAudit-2011}, neither framework has so far had a large impact on research in security.

Many machine learning problems involve inference over high-dimensional data---such as images, gene expression assays, text corpora---which can lead to both computational and statistical challenges. As a consequence, many machine learning researchers apply some form of dimensionality reduction. While some appeal to heuristic approaches such as performing principal components analysis prior to learning, theoretical foundations of projections have now become mainstream showing that learning on projected data requires less data or equivalently can be more accurate, while in some cases also requiring less time to achieve.

Projections used in machine learning are analogous to the reductions for large-scale strategic games discussed in this paper.
The Johnson-Lindenstrauss Lemma~\cite{johnson1984extensions} established the existence of low-dimensional linear projections that approximately preserve inter-point distances, a key characteristic used by many learning algorithms. The randomized version proves the same for random linear projections, %(matrices of i.i.d. zero-mean Gaussian elements),
with high probability. In their landmark paper, \cite{arriaga1999algorithmic} built on this property of random projections to show that model classes (that are robust, with some margin) are still PAC learnable~\cite{Valiant84} when randomly projected by data-independent mappings. The consequence being that learning may be possible with less data (in a formal sense) under projections. Such projections may prove useful in large-scale game reductions.

More recently, Rahimi \& Recht~\cite{rahimi2007random} use random projections to reduce the computational complexity of training support vector machine (SVM) classifiers on large data sets. For many problems, non-linear SVMs achieve state-of-the-art accuracy but take time cubic in the number of training examples $n$ to learn. SVM learning involves solving a quadratically-constrained quadratic program~\cite{scholkopf2001learning} whose dual involves the data only via inner-products---a kernel matrix. Their novel randomized projections are constructed such that inner-products are approximated uniformly, with high probability depending on the image dimension. This dimension can be taken to be much less than $n$, yielding much faster training as fewer parameters are learned, without paying with statistical performance since the SVM is relatively stable with respect to perturbations of the kernel.

% \begin{itemize}
%  \item games with large number of players, mean-field games
%  \item convex optimization for large-scale data
%  \item security games in cloud computing and adversarial learning 
%  \item projection methods
% \end{itemize}

\subsection{Contributions}

The main contributions of the paper include:
\begin{itemize}
	\item The characterization of large-scale strategic games;
 \item Reductions of large-scale strategic games using linear transformations such as random projections and their effect on equilibrium solutions;
% \item Definition of separable and submodular games as a basis for reductions in large-scale games;
 \item An analysis of convex and quadratic 2-player large-scale games and their equilibrium solutions; and
 \item An adversarial machine learning game that incorporates random projection and sampling based on
 a linear SVM formulation.
%  Two examples, a duopoly production game and a security game, that illustrate context-specific 
%  selection of linear reductions in large-scale games.
\end{itemize}

% Don't need this for a 6pg (body) paper?
%The rest of the paper is organized as follows. The next section presents the model, underlying assumptions, and the
%basic definitions related to large-scale strategic games. Section~\ref{sec:ne} studies Nash Equilibrium (NE) in Big Games
%under fairly general assumptions. A set of analytically-tractable results are discussed in
%Section~\ref{sec:quadratic} in the context of quadratic 2-player large-scale games. Section~\ref{sec:example} contains 
%two specific large-scale game examples that illustrate the concepts and provide insights on reduction schemes for large-scale games.
%The paper concludes with the remarks in Section~\ref{sec:conclusion}.

%\newpage

%%%%%%%%%%%%%%%%%%%%%%%%%%%%%%%%%%%%%%%%%%%%%%%%%%%%%%%%%%%%%%%%%%%%%%%%%%%%%%%%
\section{MODEL AND DEFINITIONS} \label{sec:model}

The general model presented in this section focuses on nonzero-sum large-scale strategic games with high-dimensional 
continuous decision spaces and reduced games obtained through a linear mapping of player decision spaces. 
%such as random projection~\cite{arriaga1999algorithmic,rahimi2007random}.

Let $\mc N:=\{P^1,P^2,\ldots,P^N\}$ be the set of players in a static, continuous kernel, $N$-Player strategic (non-cooperative) game. Each player $i \in \mc N$ chooses a pure strategy (decision) vector, $x^i$ from its convex and compact decision set $X^i \subseteq \Real^M$. The joint decision space of the game is therefore the product space $\mc X=X^1\times\ldots X^N \subseteq \Real^{M \times N}$. %, and hence  
%$X^i \subset \mc X\;\forall i.$
Each player is associated with a cost function $J^i(x^i,x^{-i}): \mc X \rightarrow \Real$, where
$x^i \in X^i$ and $x^{-i}$ is a shorthand for the decision vectors of all other players, $[x^1,\ldots,x^{i-1},x^{i+1},\ldots,x^N]$. The players are assumed to be rational and choose their decisions based on a best response 
strategy by solving the optimization problem
$$ x^{i,BR}\in\arg \min_{x^i \in X^i} J^i(x^i,x^{-i})\enspace,$$
given the actions of all other players $x^{-i}$.

Large-scale games are often identified by the fact that their players have a very large strategy or decision space.
The following straightforward definition formalizes this important distinction.

\begin{defn} \label{thm:large-scalegamedef}
Consider the static, continuous kernel, $N$-Player strategic (non-cooperative) game $\mc G(\mc N,\mc X,J)$,
where $\mc N$ is the set of players, $\mc X \subseteq \Real^{M \times N}$ is the joint decision space, and $J=[J^1,\ldots,J^N]$ denotes the real-valued player cost functions. The game is called a \textbf{large-scale strategic game}, $\mc G^B$  if the
individual player decision spaces have a very large dimension, \ie $M \gg 1$.
\end{defn}

\begin{assm} \label{thm:large-scalegameassm}
The large-scale game in Definition~\ref{thm:large-scalegamedef} is assumed to be intractable in its original form.
\end{assm}

This assumption holds in at least two well-motivated cases:
\begin{enumerate}
 \item The players cannot fully identify their and others' cost functions due to a large number of
 decision variables and preference parameters; or
 \item The large-scale problem of finding best responses or Nash equilibria is computationally infeasible within timing and resource constraints.
\end{enumerate}

As a starting point consider a reduced decision space $\mc Y \subseteq \Real^{K \times N}$, where $K<M$, obtained through a linear transformation $T: \mc X \rightarrow \mc Y$ represented by per-player transformation matrices $A^i$,
\begin{equation} \label{e:tmap}
 T: y^i=A^i x^i, \; x^i \in X^i\subseteq \Real^M, \;y^i \in Y^i\subseteq \Real^K, \;\forall i \in \mc N\enspace,
\end{equation}
as illustrated in Figure~\ref{fig:map}. Based on Assumption~\ref{thm:large-scalegameassm}, the players of the large-scale game $\mc G^B$ in Definition~\ref{thm:large-scalegamedef} make decisions on a reduced space $\mc Y$ % obtained by a linear transformation $T:\mc X \rightarrow \mc Y$,
resulting in a tractable game $\mc G^T(\mc N, \mc Y, \tilde J)$.

% Once a solution is obtained in the reduced space the players can either use a pseudo-inverse transform or use default values in making decisions on the original space. The composition of pseudo-inverse and cost functions $J$ induces cost functions $\tilde J$ on the reduced space.

The transformation matrices $A^i$ may, for example, randomly select a subset of
decision variables or data points, as defined next, or represent a \textit{random projection}.

% The linear mapping, $A$ reducing the decision space can be defined in different ways. For example, it can be random projection or a  selection matrix choosing a subset of the player decision space dimensions.

\begin{defn}[Dimension reduction through selection] \label{thm:transmat}
	A transformation matrix $A$ with $K$ rows and $M$ columns is said to be a \emph{selection matrix}, if it is of rank $K<M$ and its elements
$a_{km}$ satisfy: 
$$a_{km}\in \{0,1\} \;\text{ and } \;\sum_{m=1}^M a_{km}=1\;,\; \forall m\enspace.$$
\end{defn}

Random projections are extremely popular techniques in machine learning for dealing with the curse-of-dimensionality~\cite{arriaga1999algorithmic,rahimi2007random}.
\textit{If a random projection matrix $A_R$ is carefully chosen, then all pairwise Euclidean distances, and hence, the geometry of the set of points in $\mc X$ are preserved in $\mc Y$ with high probability. }
There are many possible constructions for the random projection matrix $A_R$ that preserve pairwise distances. The most common one is choosing entries as i.i.d. standard Gaussian random variables. Another common alternative is the random sign matrix whose entries are set to $+1$ or $-1$ with equal probability. The well-known Johnson-Lindenstrauss Lemma~\cite{johnson1984extensions} formalizes this idea. A version adopted to this paper's notation and context is presented next for completeness.
\begin{thm}\label{thm:JL-lemma}
Let $x^1, x^2 \in \mc X \subset \Real^d$ and $y^1=\frac{1}{\sqrt{r}} x^1  A_R$, $y^2=\frac{1}{\sqrt{r}} x^2 A_R $ be a pair of vectors
in $\mc X$ and their corresponding mapping to $\mc Y$. Let $A_R$ be an $d \times r$ random matrix whose entries are chosen independently from either a zero-mean unit-variance Gaussian distribution or as $+1$ or $-1$ with equal probability. Then, for $\gamma>0$,
\begin{align*}
 \nonumber P\left[ (1-\gamma)\norm{x^1-x^2}^2 \leq \norm{y^1-y^2}^2 \leq (1 + \gamma)\norm{x^1-x^2}^2 \right] \\ 
 \geq \ 1-2 e^{-(\gamma^2-\gamma^3)\frac{r}{4}}\ . \nonumber
\end{align*}%where $\gamma>0$.
\end{thm}

% \begin{defn} \label{thm:deftgame}
% Based on Assumption~\ref{thm:large-scalegameassm}, the players of the large-scale game $\mc G^B$ in Definition~\ref{thm:large-scalegamedef}
% may make decisions on a reduced space $\mc Y$ obtained by a linear transformation $T:\mc X \rightarrow \mc Y$,
% resulting in a tractable game $\mc G^T(\mc N, \mc Y, \tilde J)$.
% \end{defn}
% 
% \begin{defn} \label{thm:transmat}
% A transformation matrix $A$ with $K$ rows and $M$ columns is said to be a selection matrix if $K <M$ and its elements
% $a_{km}$ satisfy the following: 
% $$a_{km}\in \{0,1\} \;\text{ and } \;\sum_{j=1}^M a_{kj}=1\;,\; \forall k,m\enspace.$$
% \end{defn}

\begin{figure}[t!]
  \centering
  \includegraphics[width=0.9\columnwidth]{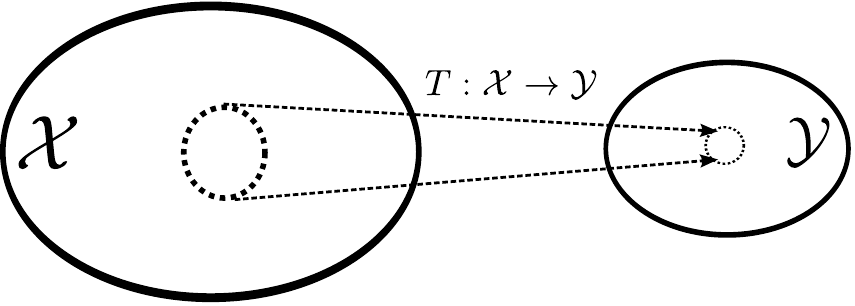}
  \caption{Illustration of the decision space transformation for large-scale strategic game reductions.}
  \label{fig:map}
\end{figure}

Note that, while it is possible to make the selection of the linear mapping $T$ and matrices $A$ introduced in the previous
section a part of the decision problem, doing so would obviously defeat the purpose of the transformation since $(A,y)$ has a higher dimension than the original large-scale game decisions, $x$.  

\begin{lem} \label{lem:combinatorial}
The combined problem of a player $i$ optimally choosing $(A^i,y^i)$ as a reduction has a higher-dimension than the 
original problem of optimally choosing $x^i$ in the large-scale game $\mc G^B$ for a 
generic linear transformation $A^i$. %Furthermore, the combined problem is mixed-integer and combinatorially explosive in the
%case when $A^i$ is a transformation matrix.
\end{lem}

\begin{proof}
The proof immediately follows from definitions. Since $A$ has $K$ rows and $M$ columns, the total number of variables
in $(A,y)$ is $K(M+1)>M$ in the reduced case, which defeats the purpose since $x^i$ has $M$ elements. %The mixed-integer problemis composed of $K$ continuous variables and the combinatorial problem of $M$ choose $K$.
\end{proof}

%%%%%%%%%%%%%%%%%%%%%%%%%%%%%%%%%%%%%%%%%%%%%%%%%%%%%%%%%%%%%%%%%%%%%%%%%%%%%%%%
\section{ANALYSIS OF LARGE-SCALE AND \\ REDUCED GAMES} \label{sec:ne}

A fundamental question of interest in large-scale games is how a Nash Equilibrium (NE) solution
of the original game, $\mc G^B$, relates to that of the tractable game $\mc G^T$ obtained through a linear transformation such as 
random projection in the decision space. First, basic
results will be discussed for the general case of $N$-players. Then, specific results for two-player
quadratic games will be presented.

\begin{lem} \label{thm:lemma}
If the cost function of player $i$, $J^i$ in a large-scale game  $\mc G^B$ is
convex in $x^i$, then $\tilde J$ of the tractable game $\mc G^T$  is
also convex in $y^i$.
\end{lem}

\begin{proof}
The result is due to linearity of the mapping $T$, see \eg \cite[Chap. 3.2.2]{boydbook}.
\end{proof}

The following well-known result from \cite{basargame} establishes sufficient conditions for existence of NE in the game $\mc G^B$.
\begin{prop} \label{thm:basarne}
If the decision space $X^i$ of each player $i$ in the large-scale game $\mc G^B$ is closed, bounded, and convex, and the
cost function $J^i$ is jointly continuous in all its arguments and strictly convex in $x^i$ for any $x^{-i}$ and
for all $i \in \mc N$, then the game admits a (pure) Nash Equilibrium solution.
\end{prop}

Combining Lemma~\ref{thm:lemma} and Proposition~\ref{thm:basarne} leads to:
\begin{prop} \label{thm:netone}
If the large-scale  game $\mc G^B$  satisfies the sufficient conditions in Proposition~\ref{thm:basarne} and admits a (pure) NE, then the tractable game $\mc G^T$ also admits a (pure) NE.
\end{prop}

%%%%%%%%%%%%%%%%%%%%%%%%%%%%%%%%%%%%%%%%%%%%%%%%%%%%%%%%%%%%%%%%%%%%%%%%%%%%%%%%
\subsection{Convex Large-Scale and Reduced Games} \label{sec:convex}

Proposition~\ref{thm:netone} leads to the question of when and where can Large-Scale and Reduced Games be equivalent. The player costs in the large-scale and reduced games $\mc G^B$, $\mc G^T$, are defined as $J_i(x^i,x^{-i})$ and 
$\tilde J_i(y^i, y^{-i})=\tilde J_i(A^i x^i,A^{-i} x^{-i})$, respectively. Assume that $J_i$ is convex in $x^i$ for all players $i$. Hence, $\mc G^T$ is also convex and both admit NE solutions from 
Lemma~\ref{thm:lemma} and Proposition~\ref{thm:basarne}. Let $x^\star$ be the NE of $\mc G^B$. The Taylor series expansions of the 
costs $J_i$ and $\tilde J_i$ around $x_i^\star$ given $x^{-i,\star}$ provide the following second-order approximations:
\begin{align}\label{e:taylor1}
\nonumber J_i(x^i,x^{-i})\approx J_i(x^{i,},x^{-i,\star})+\left[ \nabla_{x^i}J^i\right]^T \cdot(x^i-x^{i,\star}) \\ 
+\ (x^i-x^{i,\star})^T \left[ \nabla_{x^i\,x^i}^2 J^i\right]  \cdot(x^i-x^{i,\star})
\end{align}
and
\begin{eqnarray}
	&& \tilde J_i(A^i x^i,A^{-i} x^{-i}) \nonumber \\
	&\approx& \tilde J_i(A^i x^{i,\star},A^{-i} x^{-i,\star}) 
+  \left[A^{i,T}\nabla_{x^i}\tilde J^i\right]^T \cdot(x^i-x^{i,\star}) \nonumber \\
&& \ +\ (x^i-x^{i,\star})^T \left[A^{i,T} \nabla_{x^i\,x^i}^2 \tilde J^i A^i\right]  \cdot(x^i-x^{i,\star})\enspace,\label{e:taylor2}
\end{eqnarray}
where $\norm{x-x^\star}<\eps$ for a small $\eps>0$.

The convex games $\mc G^B$ and $\mc G^T$ are locally approximately equivalent around the NE $x^\star$, if 
$$ J_i(x^i,x^{-i}) = \tilde J_i(A^i x^i,A^{-i} x^{-i})+\delta$$
for $\{x \in \mc X \mid \norm{x-x^\star}<\eps\}$, 
where the scalar $\delta>0$ accounts for the discrepancy due to higher-than-second-order terms.
Then, the following relationships, obtained using basic linear algebraic manipulations, 
establish a connection between the first- and second-order terms in 
the large-scale and reduced games such that: 
\begin{equation}\label{e:cequiv1}
 \nabla_{x^i}\tilde J^i = (A^i A^{i,T})^{-1} A^i \left[\nabla_{x^i}J^i(x^i,x^{-i})\right] 
\end{equation}
and
\begin{equation}\label{e:cequiv2}
 \nabla_{x^i\,x^i}^2 \tilde J^i= (A^i A^{i,T})^{-1} A^i \left[\nabla_{x^i\,x^i}^2 J^i\right]  A^{i,T} (A^i A^{i,T})^{-1}\enspace,
 \end{equation}
for all $i$ and $x$ such that $\norm{x-x^\star}<\eps$.

Next, the relationship between large-scale and tractable games is investigated for the special case of
two-player quadratic games.

%%%%%%%%%%%%%%%%%%%%%%%%%%%%%%%%%%%%%%%%%%%%%%%%%%%%%%%%%%%%%%%%%%%%%%%%%%%%%%%%
\subsection{Quadratic 2-Player Large-Scale And Reduced Games} \label{sec:quadratic}

%Investigating Quadratic 2-Player Large-Scale Games provides a starting point for the more general cases. 
%helps further analyze the %large-scale games introduced in Section~\ref{sec:model} and illustrate the results of Section~\ref{sec:ne}.
%large-scale games and results developed above. 
Quadratic games are of particular interest in game theory as they constitute second-order approximation to games with nonlinear cost functions, while admitting closed-form equilibrium solutions that provide useful insights \cite{basargame}. They are also related to Quadratic Programming which is encountered in key machine learning training algorithms~\cite{Bishopbook}.

Consider the following special case of the game $\mc G^T$ with only two players
$P^1$ and $P^2$ having the respective cost functions:
\begin{align} 
\label{e:quadratic} \tilde{J}^1 \left(y^1,y^2\right)= y^{1,T} \tilde Q_1 y^{2}- y^{1,T} \tilde r^1 +  v^1 \enspace,\\
\label{e:quadratic2} \tilde{J}^2 \left(y^1,y^2\right)= y^{2,T} \tilde Q_2 y^{1}- y^{2,T} \tilde r^2 + v^2 \enspace.
\end{align}
%where $(\cdot)^T$ denotes vector/matrix transpose.
The game parameters are the scalars $v^1$, $v^2$, the vectors
$\tilde r^1$, $\tilde r^2$, and matrices $\tilde Q_1$, $\tilde Q_2$. %Denote this quadratic game as $\mc G^Q$.

The corresponding cost functions \eqref{e:quadratic}--\eqref{e:quadratic2} 
of the original Large-scale Game $\mc G^B$ can be written as
\begin{align} 
\label{e:quad} {J}^1 \left(x^1,x^2\right)= x^{1,T}  Q_1 x^{2}- x^{1,T}  r^1 +  v^1 \enspace,\\
\label{e:quad2} {J}^2 \left(x^1,x^2\right)= x^{2,T} Q_2 x^{1}- x^{2,T}  r^2 +  v^2 \enspace.
\end{align}
If the matrices $Q_1$ and $Q_2$ are positive definite and hence invertible~\cite{matrixcookbook}, then the cost functions $J^1$ and ${J}^2$ are both quadratic and strictly convex.
Therefore, the first derivatives vanishing serves as necessary and sufficient for optimality in calculating player best responses 
$$ x^{1,\star} = \left(Q_2 \right)^{-1} r^2 \; \text{ and }  
   x^{2,\star} = \left(Q_1 \right)^{-1} r^1\enspace. $$
Thus, $x^\star=[x^{1,\star}, x^{2,\star}]$ is the unique NE of the original large-scale game. Note that the NE strategy of one player is determined by the parameters of the other player.

\textit{When are the outcomes of these two games equivalent?} To answer this question, let
${J}^1=\tilde{J}^1$ for Player $1$. Since $y^1=A^1 x^1$ and $y^2=A^2 x^2$, the equivalence between linear
terms are immediately established by $r^1=A^{1,T} \tilde r^1$ and $r^2=A^{2,T} \tilde r^2$. Focusing on the quadratic terms,
$$ x^{1,T}  Q_1 x^{2} = y^{1,T} \tilde Q_1 y^{2}$$
and
$$ x^{1,T}  Q_1 x^{2} = x^{1,T} A^{1,T} \tilde Q_1 A^2 x^{2},$$
lead to
\begin{equation}\label{e:qequiv1}
 Q_1 = A^{1,T} \tilde Q_1 A^2.
\end{equation}
%has to hold $\forall x^1,\, x^2$. 
Multiplying each side first with $A^1$ and $A^{2,T}$, and then with $(A^1 A^{1,T})^{-1}$
and $(A^2 A^{2,T})^{-1}$ from left and right, respectively, yields
\begin{equation}\label{e:qequiv2}
 \tilde Q_1= (A^1 A^{1,T})^{-1} A^1  Q_1 A^{2,T} (A^2 A^{2,T})^{-1}.
 \end{equation}
The analysis can be repeated similarly for  ${J}^2=\tilde{J}^2$.
% where $Q_1=A^{1,T} \tilde Q_1 A^2$, $Q_2=A^{2,T} \tilde Q_2 A^1$, $r^1=A^{1,T} \tilde r^1$, and
% $r^2=A^{2,T} \tilde r^2$.

From (\ref{e:qequiv1}) and (\ref{e:qequiv2}), it is observed that unless $A^1=A^2$, \ie the players
use the same reduction mappings, positive definiteness of $Q^i$ does not guarantee the positive definiteness
of $\tilde{Q}^i$ and vice versa. However, if $A^1=A^2$, then one matrix being positive definite ensures that
the other one is so as well. In this case, the matrices can be decomposed as
$R^i R^{i,T}$ and  $\tilde R^i \tilde R^{i,T}$, respectively. Consequently,
$$ \tilde R^i \tilde R^{i,T}=(A A^{T})^{-1} A R^i R^{i,T}  A^{T} (A A^{T})^{-1},\; i=1, 2, $$
or
\begin{equation}\label{e:qequiv3}
\tilde R^i= (A A^{T})^{-1} A R^i ,\; i=1, 2.
\end{equation}
The relationships (\ref{e:qequiv1})-(\ref{e:qequiv3}) establish a connection between the quadratic terms in 
the large-scale and reduced versions of the game. These can be used in design of reductions and/or choice of cost parameters.

% From \eqref{e:quad}--\eqref{e:quad2} and the fact that $y^\star=[y^{1,\star}, y^{2,\star}]$ is the unique NE of $\mc G^Q$,
% $x^\star$, the NE of the game in the original decision space $\mc X$ is the solution of
% $$ Q^2 x^{1,\star}=r^2 \text{ and }  Q^1 x^{2,\star}=r^1 \enspace.$$
% Note that this is a special case of Theorem~\ref{thm:inverse}.
% 
% Similarly, if the matrices $Q^1$ and $Q^2$ are invertible, then there is a one-to-one correspondence
% between the NE 
% $$x^{1,\star}=\left( Q_2 \right)^{-1}r^2, \;\; x^{2,\star}=\left( Q_1 \right)^{-1}r^1$$ 
% of the original and tractable versions of the quadratic game such that
% $$  y^{1,\star} = A^1 x^{1,\star} \text{ and } y^{2,\star} = A^2 x^{2,\star} \enspace,$$
% which directly follows from the definitions above:
% \begin{align*}
%  y^{1,\star}=A^1 x^{1,\star} &=  A^1 \left( Q_2 \right)^{-1}r^2  \\
%   &=  A^1 \left(A^1\right)^{-1} \left(\tilde Q_1 \right)^{-1}
% \left(A^{1,T}\right)^{-1} A^{1,T} \tilde r^1 \\
%   &= \left(\tilde Q_1 \right)^{-1} \tilde r^1\enspace.
% \end{align*}
% This result is a special case of Theorem~\ref{thm:fwd}.

%%%%%%%%%%%%%%%%%%%%%%%%%%%%%%%%%%%%%%%%%%%%%%%%%%%%%%%%%%%%%%%%%%%%%%%%%%%%%%%%
\section{LARGE-SCALE GAMES FOR ADVERSARIAL MACHINE LEARNING} \label{sec:advml}

Security games have been used increasingly to model decision making in network and real-life problems with resource constraints~\cite{alpcan-book,tambegame,Barth-2011-reactive}. Adversarial machine learning (AdvML) is
the study of effective machine (or statistical) learning techniques against an adversarial opponent, who aims to disrupt the learning
and hence subsequent decision making process with malicious intent. Many adversarial learning problems
can be posed as security games. Moreover, high-dimensional and high-volume data generated by modern systems naturally
lead to large-scale game formulations which can be analyzed adopting an approach similar to the one discussed
in the previous sections. A specific adversarial learning problem based on linear Support Vector Machines (SVMs)
is investigated next, which leads to a large-scale game formulation.

%\subsection{Large-Scale Games for Adversarial Machine Learning (AdvML) using Linear Support Vector Machines (SVMs)}

Consider a linear SVM as a binary classifier trained using a large and high-dimensional 
labeled data set consisting of $n$ $d-$dimensional real data vectors %$x^d_i \in \Real^d$, 
with respective $\{-1, +1\}$ labels, where $n, d \gg 1$. The choice of linear SVM is without loss of any
generality since nonlinear kernels can be embedded into the random projection, i.e. the inner product of
projected points can approximate their original kernel evaluation if the transformation is 
carefully selected~\cite{rahimi2007random}.

%$y^d_i \in \{-1, +1\}$ for $i=1,\ldots,n$, where $n, d >>1$. 
The training of the SVM involves solving an optimization problem
where a hyperplane with normal vector, $w^\star$, is obtained that maximizes the (soft) geometric margin (the minimum distance
of a data point to the hyperplane). The dual formulation of the problem leads to the following well-known formulation:
\begin{align}\label{e:svmbasic}
 \nonumber \max_{\a} \one^T \a - \dfrac{1}{2} \a^T Y X X^T Y \a \\
 \rm{s.t.}\; \one^T Y \a=0, \quad 0\leq \a \leq C,
\end{align}
where $\a$ is the vector of Lagrange multipliers, $C>0$ is a constant, $\one$ is a vector of ones, $X \in \Real^{n\times d}$ is an $n\times d$ matrix whose rows are the data vectors, and $Y \in \Real^{n\times n}$ is a diagonal matrix with diagonal entries the corresponding $\{-1, +1\}$ labels. The optimal separating hyperplane is represented by $w^\star=\a^{\star,T} Y X$, where $\a^\star$ is the optimal solution to (\ref{e:svmbasic}). Note that, $Y^T=Y$ by definition and if an $\a_i^\star>0$ then the corresponding data point is called a \emph{support vector}. 

Let the player who solves (\ref{e:svmbasic}) with the aim of maximizing the geometric (soft) margin be called ``Defender''. Since
the data set has a large number of points, $n\gg 1$, and each data point is high-dimensional, $d\gg 1$, the Defender adopts
random projection as dimension reduction and random data selection as volume reduction strategies. Let $A_R$ be a $d \times r$ 
random projection matrix, $r<d$, as defined in Theorem~\ref{thm:JL-lemma} and $A_S$ be a $n\times n$ matrix, which 
extends the selection matrix in Definition~\ref{thm:transmat} by adding zero rows to appropriate places. Thus, $X A_R$ is a random projection of data vectors where their distances are preserved with a high probability.
The selection mapping $A_S X$, on the other hand, deletes a subset of data vectors to reduce the data volume. The reduction
in training data dimension and volume decreases the computational and information burden of the Defender. It also provides a certain
degree of robustness against malicious attacks in an adversarial setting.

A common attack type in adversarial learning is malicious distortion of training data points (vectors) by an Attacker. Let $X+D$
be the distorted data matrix, where $D$ is determined by the Attacker. The amount of injected distortion to the data is often bounded due to either the computational burden or increasing risk of discovery of the Attack(er). For example, some of the rows of $D$ may be zero indicating that the Attacker distorts only a subset of the data.

The counterpart of (\ref{e:svmbasic}) in the adversarial learning formulation describes is then:
\begin{align}\label{e:svmadv}
 \nonumber \max_{\a} \one^T A_S \a - \dfrac{1}{2} \a^T Y A_S (X+D) A_R A_R^T (X+D)^T A_S Y \a \\
 \rm{s.t.}\; \one^T Y A_S \a=0, \quad 0\leq \a \leq C,
\end{align}

The optimal geometric margin $\g^\star$ of the canonical hyperplane $w^\star$ obtained from (\ref{e:svmbasic}) is defined as $\g^\star=1/\norm{w^\star}_2$, where $\norm{w^\star}_2^2=\sum_i \a_i$. Define, likewise, $\tilde \g^\star$ using ~(\ref{e:svmadv}). 
Since geometric margin plays a foundational role in the formulation of the binary SVM classification problem, it makes sense
to use it as a criterion in the adversarial learning game. Hence, the player objectives can be posed as minimizing (maximizing)
the distortion in the geometric margin $\abs{\g^\star-\tilde \g^\star}$ for the Defender (Attacker), respectively.

As a starting point of the analysis, the optimization problem (\ref{e:svmadv}) is reformulated by capturing the impact
of random selection mapping $A_S$ through a set of new constraints:
\begin{align}\label{e:svmadvns}
 \nonumber \max_{\a} \one^T \a - \dfrac{1}{2} \a^T Y (X+D) A_R A_R^T (X+D)^T  Y \a \\
 \rm{s.t. } \quad \one^T Y \a=0, \quad 0\leq \a \leq C, \\
 \rm{and }\; (I-A_S)X=0, \quad  (I-A_S)\a=0. \label{e:svmadvnscon}
\end{align}

Note that, 
$$\a^T Y (X+D) A_R A_R^T (X+D)^T  Y \a = \norm{\a^T Y (X+D) A_R}_2^2\ .$$
Using the triangle inequality,
\begin{eqnarray}\label{e:normbound}
	(1-\delta) \norm{\a^T Y X A_R}_2^2 &\leq& \norm{\a^T Y (X+D) A_R}_2^2 \\
			&\leq&  (1+\delta) \norm{\a^T Y X A_R}_2^2\ , \nonumber
\end{eqnarray}
where
$$ \delta:= \dfrac{\norm{\a^T Y D A_R}_2^2}{\norm{\a^T Y X A_R}_2^2} <1\ .$$

Let $Z_{pds}(\a_{pds}^\star)$ be the optimal value of (\ref{e:svmadv}) and (\ref{e:svmadvns})-(\ref{e:svmadvnscon}). The value, $Z_{pd}(\a_{pd}^\star)$, obtained by resolving (\ref{e:svmadvns}) without the constraints in (\ref{e:svmadvnscon}) clearly leads to a higher or equal value such that 
$$Z_{pd}=\beta Z_{pds}\ , \quad \beta \geq 1\ . $$
It is worth noting that $\beta$ is a function of the data $X$ as well as $A_R$, $A_S$, and $D$.

Define 
$$Z_p(\a_{pd}^\star):=\one^T \a_{pd}^\star - \norm{\a_{pd}^{\star,T} Y X A_R}_2^2\ ,$$
which is the optimal value without any malicious distortion, $D=0$, of the data set.
Then, from (\ref{e:normbound}), 
% $$ \norm{\a^T Y X A_R}_2^2 \leq \dfrac{1}{1-\delta} \norm{\a^T Y (X+D) A_R}_2^2,$$
% and hence,
$$Z_p(\a_{pd}^\star) \geq Z_{pd}(\a_{pd}^\star)-\delta \norm{\a_{pd}^{\star,T} Y X A_R}_2^2 $$
and hence
\begin{equation}\label{e:ineq1}
 Z_p(\a_{pd}^\star) \geq \beta Z_{pds}(\a_{pds}^\star)-\delta \norm{\a_{pd}^{\star,T} Y X A_R}_2^2\ .
\end{equation}

Let $Z(\a^\star)$ denote the optimal value of the original problem, (\ref{e:svmbasic}).
It is important to note that $\a_{pd}^\star$ is also a feasible (but clearly not optimal) solution of (\ref{e:svmbasic}).
Hence, $Z(\a^\star)\geq Z_p(\a_{pd}^\star)$ by definition.

The rest of the analysis closely follows the one in~\cite{linearsvmproj}.
Let $V \in \Real^{d \times \rho}$ be any matrix with orthonormal columns. Define
$E:=V^T V - V^T A_R A_R^T V$, and assume $\norm{E}_2 < \phi$ for a given $A_R$.
Then,
$$ Z(\a^\star)\geq Z_p(\a_{pd}^\star) - \dfrac{1}{2} \norm{E}_2 \norm{\a_{pd}^{\star,T} Y X}_2^2\ .$$
It is shown in~\cite{linearsvmproj} that 
$$ \norm{\a^{\star,T}_{pd} Y X }_2^2 \leq \dfrac{1}{1-\norm{E}_2} \norm{\a^{\star,T}_{pd} Y X A_R}_2^2\ .$$
Thus, from (\ref{e:ineq1}), 
\begin{align}\label{e:ineq2}
Z(\a^\star) \geq & \; \beta Z_{pds}(\a_{pds}^\star)  \nonumber \\
 & -\dfrac{\beta}{2} \left( \dfrac{\norm{E}_2}{1-\norm{E}_2}+2\delta \right)
 \norm{\a_{pds}^{\star,T} Y X A_R}_2^2\ . 
\end{align}

Remember that, from its definition $w^\star=\a^{\star,T} Y X=\sum_i \a^\star_i$ and 
$w^{\star}_{pds}=\a_{pds}^{\star,T} Y X=\sum_i \a_{pds,i}^{\star}$. Therefore,
$Z(\a^\star)=0.5 \norm{w^\star}_2^2$, $Z_{pds}=0.5 \norm{w^\star_{pds}}_2^2$, and
the geometric margins are $\g^\star=1/\norm{w^\star}_2$ and $\tilde \g^\star=1/\norm{w^{\star}_{pds}}_2$.
Combining these definitions with (\ref{e:ineq2}) leads to
$$ Z(\a^\star)\geq  \beta\left(1- \dfrac{\norm{E}_2}{1-\norm{E}_2}-2\delta \right)Z_{pds}(\a_{pds}^\star)$$
or
\begin{equation}\label{e:svmresult}
 \tilde \g^{\star 2} \geq \beta \left(1-\phi-2\delta\right) \g^{\star 2}\ .
\end{equation}

Now, the adversarial machine learning game between the Defender and Attacker can be defined
based on this worst-case gap between the margins, 
$$ \abs{\g^{\star 2}-\tilde \g^{\star 2}}/\g^{\star 2} \leq 1-\beta(1-\phi-2\delta)\ .$$
%where it is assumed that $\beta(1-\phi-2\delta\right)$.
Let $(A_R,A_S)$ represent the actions of the Defender and $D$ of the Attacker. As one possibility, the Defender aims to maximize $\beta(1-\phi-2\delta)$ to decrease the worst-case margin gap, while the Attacker tries to minimize it. It is assumed here that random projection, selection,
and malicious distortions do not inadvertently increase the margin and help the Defender. 
The cost functions also capture the computational
gains for the Defender due to reductions and risk of detection for the Attacker.
Thus, the cost functions of the Defender and Attacker are defined as:
\begin{align}
 \min_{A_R,A_S} \, J^D(A_R,A_S,D)=-\beta(A_S)(1-\phi(A_R)-2\delta(D)) \label{e:advmldefobj} \\
 + c^D_R \norm{A_R} + c^D_S \norm{A_S} , \nonumber \\
 \min_{D} \, J^A(A_R,A_S,D)= \beta(A_S)(1-\phi(A_R)-2\delta(D))  \nonumber \\
 +c^A \norm{D} \label{e:advmlattobj}
\end{align}

The following observations can be made:
\begin{itemize}
 \item $\beta(A_S)$ is increasing in $\norm{I-A_S}$, i.e. the number of samples deleted from the data set;
 \item $\phi(A_R)$ is increasing as the random projection space decreases in number of dimensions, \ie
 number of columns, $r$, in the $d \times r$ matrix $A_R$ decreases; and
 \item $\delta(D)$ is increasing in $\norm{D-I}$, \ie the amount of distortion introduced to training data
 by malicious attacker increases.
\end{itemize}

The adversarial machine learning game is then defined as 
$$\mc G^{AdvML}\left(\{Defender, Attacker\},\{(A_R,A_S),\, D\},\{J^D,\, J^A\}\right),$$
where $J^D$ and $J^A$ are defined in (\ref{e:advmldefobj}), (\ref{e:advmlattobj}), respectively.
It is important to note that this game can only be solved numerically due to the nonlinear nature of 
the functions $\beta(A_S)$, $\phi(A_R)$, and $\delta(D)$.

\section{CONCLUSION} \label{sec:conclusion}

A framework for large-scale strategic games with continuous decision variables has been introduced in this paper. First, a characterization and basic definitions of large-scale strategic games have been presented. Second, motivated by information limitations, reduction of large-scale strategic games using linear transformations such as random projections and their effect on equilibrium solutions have been discussed. Third, a set of analytical results on convex and quadratic 2-player large-scale games and their equilibrium solutions have been obtained. Finally, a specific adversarial machine learning game formulation has been used to illustrate context-specific selection of linear reductions in large-scale games.

Large-scale strategic games as defined in this paper can be seen as complementary to their finite state and/or action counterparts, and hence to the game abstraction literature. There are multiple interesting future research directions in the continuous-kernel game domain. One open direction is further investigation of specific transformation techniques for reduction of large-scale games and projection methods. A second direction is exploration of practical solution algorithms for large-scale data-driven games and learning schemes, e.g. the game defined at the end of Section~\ref{sec:advml}. A third direction is additional applications of adversarial learning games to
specific problem domains.

\section{ACKNOWLEDGMENTS}

The authors thank Dr. Sarah Monazam Erfani for the helpful comments and discussions.

%%%%%%%%%%%%%%%%%%%%%%%%%%%%%%%%%%%%%%%%%%%%%%%%%%%%%%%%%%%%%%%%%%%%%%%%%%%%%%%%
\bibliographystyle{IEEE}
\bibliography{biggame}

\begin{thebibliography}{10}

\bibitem{mechbook}
T.~Alpcan, H.~Boche, M.L. Honig, and H.V. Poor,
\newblock {\em {Mechanisms and Games for Dynamic Spectrum Allocation}},
\newblock Cambridge University Press, December 2013.

\bibitem{debsawind}
D.~Chattopadhyay and T.~Alpcan,
\newblock ``{A Game-Theoretic Analysis of Wind Generation Variability on
  Electricity Markets},''
\newblock {\em IEEE Trans. on Power Systems}, vol. 29, no. 5, pp. 2069--2077,
  September 2014.

\bibitem{walidsmartgridgame}
W.~Saad, Zhu Han, H.V. Poor, and T.~Basar,
\newblock ``{Game-Theoretic Methods for the Smart Grid: An Overview of
  Microgrid Systems, Demand-Side Management, and Smart Grid Communications},''
\newblock {\em IEEE Signal Proc. Magazine}, vol. 29, no. 5, pp. 86--105,
  September 2012.

\bibitem{lawtansurisk}
Y.W. Law, T.~Alpcan, and M.~Palaniswami,
\newblock ``{Security Games for Risk Minimization in Automatic Generation
  Control},''
\newblock {\em IEEE Trans. on Power Systems}, vol. 30, no. 1, pp. 223--232,
  January 2015.

\bibitem{alpcan-book}
T.~Alpcan and T.~Ba\c{s}ar,
\newblock {\em {Network Security: A Decision and Game Theoretic Approach}},
\newblock Cambridge University Press, 2011.

\bibitem{tambegame}
M.~Tambe,
\newblock {\em {Security and Game Theory: Algorithms, Deployed Systems, Lessons
  Learned}},
\newblock Cambridge University Press, 2011.

\bibitem{Barth-2011-reactive}
Adam Barth, Benjamin I.~P. Rubinstein, Mukund Sundararajan, John~C. Mitchell,
  Dawn Song, and Peter~L. Bartlett,
\newblock ``{A Learning-Based Approach to Reactive Security},''
\newblock {\em IEEE Trans. Dependable and Secure Comp.}, vol. 9, no. 4, pp.
  482--493, 2012,
\newblock Special Issue on Learning, Games and Security.

\bibitem{basarcyberphysical1}
Quanyan Zhu and T.~Basar,
\newblock ``{Game-Theoretic Methods for Robustness, Security, and Resilience of
  Cyberphysical Control Systems: Games-in-Games Principle for Optimal
  Cross-Layer Resilient Control Systems},''
\newblock {\em IEEE Control Systems}, vol. 35, no. 1, pp. 46--65, February
  2015.

\bibitem{convexbigdata1}
V.~Cevher, S.~Becker, and M.~Schmidt,
\newblock ``{Convex Optimization for Big Data: Scalable, randomized, and
  parallel algorithms for big data analytics},''
\newblock {\em IEEE Signal Proc. Magazine}, vol. 31, no. 5, pp. 32--43, 2014.

\bibitem{AAAI1510039}
Tuomas Sandholm,
\newblock ``{Abstraction for Solving Large Incomplete-Information Games},''
\newblock in {\em {AAAI}}, 2015.

\bibitem{algogame}
Noam Nisan, Tim Roughgarden, Eva Tardos, and Vijay~V. Vazirani,
\newblock {\em {Algorithmic Game Theory}},
\newblock Cambridge University Press, New York, NY, USA, 2007.

\bibitem{tambelargegame1}
B.~Bosansky, A.~Jiang, M.~Tambe, and C.~Kiekintveld,
\newblock ``{Combining Compact Representation and Incremental Generation in
  Large Games with Sequential Strategies},''
\newblock in {\em {AAAI}}, January 2015, pp. 812--818.

\bibitem{wardrop1}
A.~Haurie and P.~Marcotte,
\newblock ``{On the relationship between Nash---Cournot and Wardrop
  equilibria},''
\newblock {\em Networks}, vol. 15, no. 3, pp. 295--308, 1985.

\bibitem{tembinemf1}
H.~Tembine,
\newblock ``{Nonasymptotic Mean-Field Games},''
\newblock {\em IEEE Trans. on Cybernetics}, vol. 44, no. 12, pp. 2744--2756,
  December 2014.

\bibitem{GreeneKLS09}
Kshanti~A. Greene, Joe~Michael Kniss, George~F. Luger, and Carl~R. Stern,
\newblock ``{Satisficing the Masses: Applying Game Theory to Large-Scale,
  Democratic Decision Problems},''
\newblock in {\em {12th {IEEE} Intl. Conf. Comp Sci. Eng. {(CSE)}}}, Vancouver,
  BC, Canada, August 2009, pp. 1156--1162.

\bibitem{MarcolinoXJTB14}
Leandro~Soriano Marcolino, Haifeng Xu, Albert~Xin Jiang, Milind Tambe, and Emma
  Bowring,
\newblock ``{Give a Hard Problem to a Diverse Team: Exploring Large Action
  Spaces},''
\newblock in {\em {AAAI}}, 2014, pp. 1485--1491.

\bibitem{YangJTO13}
Rong Yang, Albert~Xin Jiang, Milind Tambe, and Fernando Ord{\'o}{\~n}ez,
\newblock ``{Scaling-up Security Games with Boundedly Rational Adversaries: {A}
  Cutting-plane Approach},''
\newblock in {\em {IJCAI}}, August 2013, pp. 404--410.

\bibitem{Barreno-2010}
Marco Barreno, Blaine Nelson, Anthony~D. Joseph, and J.~D. Tygar,
\newblock ``{The Security of Machine Learning},''
\newblock {\em Machine Learning}, vol. 81, no. 2, pp. 121--148, November 2010.

\bibitem{lindell2000privacy}
Yehuda Lindell and Benny Pinkas,
\newblock ``{Privacy preserving data mining},''
\newblock in {\em {CRYPTO}}, 2000, pp. 36--54.

\bibitem{Dwork06}
Cynthia Dwork, Frank McSherry, Kobbi Nissim, and Adam Smith,
\newblock ``{Calibrating Noise to Sensitivity in Private Data Analysis},''
\newblock in {\em {TCC}}, 2006, pp. 265--284.

\bibitem{wittel-wu-2004-attacking}
Gregory~L. Wittel and S.~Felix Wu,
\newblock ``{On Attacking Statistical Spam Filters},''
\newblock in {\em {CEAS}}, 2004.

\bibitem{lowd-meek-2005-good}
Daniel Lowd and Christopher Meek,
\newblock ``{Good Word Attacks on Statistical Spam Filters},''
\newblock in {\em {CEAS}}, 2005.

\bibitem{LEET2008}
Blaine Nelson, Marco Barreno, Fuching~Jack Chi, Anthony~D. Joseph, Benjamin
  I.~P. Rubinstein, Udam Saini, Charles Sutton, J.~D. Tygar, and Kai Xia,
\newblock ``{Exploiting machine learning to subvert your spam filter},''
\newblock in {\em {Proc. 1st {USENIX} Work. Large-Scale Exploits and Emergent
  Threats ({LEET})}}. 2008, pp. 1--9, USENIX Association.

\bibitem{newsome2006paragraph}
James Newsome, Brad Karp, and Dawn Song,
\newblock ``{Paragraph: Thwarting Signature Learning by Training
  Maliciously},''
\newblock in {\em {RAID}}, 2006, pp. 81--105.

\bibitem{fogla-lee-2006-evading}
Prahlad Fogla and Wenke Lee,
\newblock ``{Evading Network Anomaly Detection Systems: Formal Reasoning and
  Practical Techniques},''
\newblock in {\em {CCS}}, 2006, pp. 59--68.

\bibitem{IMC09}
Benjamin I.~P. Rubinstein, Blaine Nelson, Ling Huang, Anthony~D. Joseph,
  Shing-hon Lau, Satish Rao, Nina Taft, and J.~D. Tygar,
\newblock ``{{ANTIDOTE}: Understanding and Defending against Poisoning of
  Anomaly Detectors},''
\newblock in {\em {IMC}}, 2009, pp. 1--14.

\bibitem{Dalvi-Domingos-KDD-2004}
Nilesh Dalvi, Pedro Domingos, Mausam, Sumit Sanghai, and Deepak Verma,
\newblock ``{Adversarial Classification},''
\newblock in {\em {KDD}}, 2004, pp. 99--108.

\bibitem{Bruckner-Scheffer-NIPS-2009}
Michael Br{\"u}ckner and Tobias Scheffer,
\newblock ``{Nash Equilibria of Static Prediction Games},''
\newblock in {\em {NIPS}}, 2009, pp. 171--179.

\bibitem{Cesa-Bianchi-Lugosi-ExpertGames-2006}
Nicol{\`o} Cesa-Bianchi and G{\'a}bor Lugosi,
\newblock {\em {Prediction, Learning, and Games}},
\newblock Cambridge University Press, New York, NY, USA, 2006.

\bibitem{Blocki-RegretAudit-2011}
Jeremiah Blocki, Nicolas Christin, Anupam Datta, and Arunesh Sinha,
\newblock ``{Regret Minimizing Audits: A Learning-Theoretic Basis for Privacy
  Protection},''
\newblock in {\em {CSF}}, 2011, pp. 312--327.

\bibitem{Huber-RS-1981}
Peter~J. Huber,
\newblock {\em {Robust Statistics}},
\newblock {Probability and Mathematical Statistics}. John Wiley and Sons, New
  York, NY, USA, 1981.

\bibitem{johnson1984extensions}
William~B Johnson and Joram Lindenstrauss,
\newblock ``{Extensions of {L}ipschitz mappings into a {H}ilbert space},''
\newblock {\em Contemporary Mathematics}, vol. 26, no. 189-206, pp. 189--206,
  1984.

\bibitem{arriaga1999algorithmic}
Rosa~I Arriaga and Santosh Vempala,
\newblock ``{An algorithmic theory of learning: Robust concepts and random
  projection},''
\newblock in {\em {FOCS}}, 1999, pp. 616--623.

\bibitem{Valiant84}
Leslie Valiant,
\newblock ``{A theory of the learnable},''
\newblock {\em Communications of the ACM}, vol. 27, pp. 1134--1142, 1984.

\bibitem{rahimi2007random}
Ali Rahimi and Benjamin Recht,
\newblock ``{Random features for large-scale kernel machines},''
\newblock in {\em {Advances in neural information processing systems}}, 2007,
  pp. 1177--1184.

\bibitem{scholkopf2001learning}
Bernhard Scholkopf and Alexander~J Smola,
\newblock {\em {Learning with kernels: support vector machines, regularization,
  optimization, and beyond}},
\newblock MIT press, 2001.

\bibitem{boydbook}
S.~Boyd and L.~Vandenberghe,
\newblock {\em {Convex Optimization}},
\newblock Cambridge University Press, 2004.

\bibitem{basargame}
T.~Ba\c{s}ar and G.~J. Olsder,
\newblock {\em {{D}ynamic Noncooperative Game Theory}},
\newblock SIAM, 2nd edition, 1999.

\bibitem{Bishopbook}
Christopher~M. Bishop,
\newblock {\em {Pattern Recognition and Machine Learning}},
\newblock Springer-Verlag, 2006.

\bibitem{matrixcookbook}
K.~B. Petersen and M.~S. Pedersen,
\newblock ``{The Matrix Cookbook},'' October 2008.

\bibitem{linearsvmproj}
C.~Boutsidis M. Magdon-Ismail P.~Drineas {S. Paul},
\newblock ``{Random Projections for Linear Support Vector Machines},''
\newblock {\em ACM Transactions on Knowledge Discovery from Data}, vol. 8, no.
  4, pp. 22:1--22:25, Oct. 2014.

\end{thebibliography}

\end{document}